\documentclass[runningheads]{llncs}

\usepackage{amsmath}
\usepackage{amssymb}

\newcommand {\abs}[1]{\left\vert#1\right\vert}
\newcommand {\set}[1]{\left\{#1\right\}}

\newcommand {\defined} {\stackrel{def} {=}}

\usepackage{epsfig}
\usepackage{algorithm}
\usepackage{algpseudocode}
\usepackage{algpascal}
\usepackage{multirow}
\usepackage{graphicx}
\usepackage{subfigure}
\usepackage{color}
\definecolor{red}{rgb}{1,0,0}
\definecolor{blue}{rgb}{0,0,1}
\definecolor{yellow}{cmyk}{0,0,1,0}

\newcommand{\runningtitle}[1]{\vspace{0.5ex}\noindent{{\small \textbf{\boldmath #1}}}}

\newcommand{\comment}[1]{}

\newcommand{\prb}{\textsc{Rlp}}
\newcommand{\prbmax}{\textsc{Pmp}}

\newcommand {\pp}   {{\cal P}}

\newcommand {\instrlp} {(G, \pp, d)}

\newcommand {\cost} {\textit{cost}}

\newcommand {\reg} {\textit{reg}}

\newcommand{\alg}{\textsc{Alg}}
\newcommand{\nominator}{\log (\abs{E}/d) \cdot \log d}
\newcommand{\lb}{\left(\frac{\nominator}{\log (\nominator)}\right)}
\begin{document}

\title{Online Regenerator Placement\thanks{This work was supported in part by the Israel Science Foundation grant No. 1249/08 and British Council Grant UKTELHAI09.} \thanks{A preliminary version of this paper appeared in the \emph{15th  International Conference on Principles of Distributed Systems (OPODIS), Toulouse, France, December 2011, pp.~4--17.}}}
\titlerunning{Online Regenerator Placement}

\author{George B. Mertzios~\inst{1}, Mordechai Shalom\thanks{Corresponding Author: Mendel Singer 15/6, 34984, HAIFA-ISRAEL, Tel: +972(54)753 6550}~\inst{2},\\
Prudence W.H. Wong~\inst{3}, Shmuel Zaks~\inst{4}}
\authorrunning{G. B. Mertzios, M. Shalom, P.W.H. Wong and S. Zaks}

\institute{
School of Engineering and Computing Sciences, Durham University, UK
\email{george.mertzios@durham.ac.uk}
\and
TelHai College, Upper Galilee, 12210, Israel. \email{cmshalom@telhai.ac.il}
\and
Department of Computer Science, University of Liverpool, Liverpool, UK.
\email{pwong@liverpool.ac.uk}
\and
Department of Computer Science, Technion, Haifa, Israel.
\email{zaks@cs.technion.ac.il}
}

\maketitle
\begin{abstract}
Connections between nodes in optical networks are realized by lightpaths.
Due to the decay of the signal, a regenerator has to be placed on every lightpath
after at most $d$ hops, for some given positive integer $d$.
A regenerator can serve only one lightpath.
The placement of regenerators has become an active area of research during recent years,
and various optimization problems have been studied.
The first such problem is the Regeneration Location Problem ($\prb$),
where the goal is to place the regenerators so as to
minimize the total number of nodes containing them.
We consider two extreme cases of online $\prb$ regarding
the value of $d$ and the number $k$ of regenerators
that can be used in any single node.
(1) $d$ is arbitrary and $k$ unbounded. In this case a feasible solution always exists.
We show an $O(\log \abs{X} \cdot \log d)$-competitive randomized algorithm
for any network topology, where $X$ is the set of paths of length $d$.
The algorithm can be made deterministic in some cases.
We show a deterministic lower bound of $\Omega \lb$, where $E$ is the edge set.
(2) $d=2$ and $k=1$. In this case there is not necessarily a solution
for a given input. We distinguish between feasible inputs (for which there is a solution)
and infeasible ones.
In the latter case, the objective is to satisfy the maximum number of lightpaths.
For a path topology we show a lower bound of $\sqrt{l}/2$ for the competitive ratio
(where $l$ is the number of internal nodes of the longest lightpath)
on infeasible inputs,
and a tight bound of $3$ for the competitive ratio on feasible inputs.
\end{abstract}


\setcounter{page}{0}
\newpage

\section{Introduction}\label{sec:intro}

\runningtitle{Background}
\label{sec:background}
Optical wavelength-division multiplexing (WDM) is the most promising technology today that enables us to deal with the enormous growth of traffic in communication networks, like the Internet. Optical fibers using WDM technology can carry around 80 wavelengths (colors) in real networks and up to few hundreds in testbeds. As satisfactory solutions have been found for various coloring problems, the focus of studies shifts from the number of colors to the hardware cost. These new measures provide better understanding for designing and routing in optical networks.

A communication between a pair of nodes is done via a \emph{lightpath}.
The energy of the signal along a lightpath decreases and thus amplifiers are used every fixed distance. Yet, as the amplifiers introduce noise into the signal there is a need to place a regenerator every at most $d$ hops.

There is a limit imposed by the technology on the number of regenerators that can be placed in a network node \cite{CLR10,FMMMZ11}. We denote this limit by $k$ and refer to the case where this limit is not likely to be reached by any regenerator placement as $k=\infty$.

\runningtitle{The problems}
Given a network $G$, a set of lightpaths in $G$, and integers $d$ and $k$, we need to place regenerators at the nodes of the network, such that a) for each lightpath there is a regenerator in at least one of each $d$ consecutive internal nodes, and b) at most $k$ regenerators are placed at any node. When $k=\infty$ we consider the \emph{regenerator location problem} (\prb) where the objective is to minimize the number of nodes that are assigned regenerators. When $k$ is bounded there are inputs for which there is no feasible regenerator placement that satisfies both conditions. For example, consider the case $d=2$ and $k=1$, and three identical lightpaths $u-v-w-x$. Each of   these lightpaths must have a regenerator either at $v$ or $w$, and this is clearly impossible. In this case we consider the \emph{Path Maximization Problem} (\prbmax) that seeks for regenerator placements that serve as many lightpaths as possible. We consider online algorithms (see~\cite{BE98}) for these problems.

\runningtitle{Online algorithms}
In the online setting the lightpaths are given one at a time, the algorithm has to decide on the locations of the regenerators and cannot change the decision later.  An algorithm is \emph{$c$-competitive} for $\prb$, for $c \geq 1$, if for every input the number of locations used is no more than $c$ times the locations used by an optimal offline algorithm. An online algorithm is \emph{$c$-competitive} for $\prbmax$, for $c\geq 1$, if the number of lightpaths that it satisfies is at least $1/c$ times the number of lightpaths that could be satisfied by an optimal offline algorithm.

\runningtitle{Related Work}
Placement of regenerators in optical networks has become an active area in recent years. Most of the researches have focused on the technological aspects of the problems. Moreover, heuristics and simulations heve been performed in order to reduce the number of regenerators are performed in  (e.g., ~\cite{CLR10,FGG+10a,KS01,PPK08,SGSG04,YR02,YR05}).
The regenerator location problem ($\prb$) was shown to be NP-complete in \cite{CLR10}, followed by heuristics and simulations. In \cite{FMMMZ11} theoretical results for the offline version of $\prb$ are presented. The authors study four variants of the problem, depending on whether the number $k$ of regenerators per node is bounded, and whether the routings of the requests are given. Regarding the complexity of the problem, they present polynomial-time algorithms and NP-completeness results for a variety of special cases.

We note that while considering the path topology, $\prb$ has implications for the following scheduling problem:
Assume a company has $n$ cars and that car $i$ needs to be serviced within
every at most $d$ days between day $a_i$ and $b_i$.
Furthermore, assume that the garage can serve at most $k$ cars per day
and charges a certain cost each time the garage is used.
The objective is to service the cars in the fewest number of days and hence minimizing the number of times the garage is used.

Other objective functions have also been considered in the context of regenerator placement.
E.g., in \cite{MSSZ12} the problem of minimizing the total number of regenerators
is studied under other settings.

\runningtitle{Our Contribution}
In this paper we study the online version of the regenerator location problem, and consider two extreme cases regarding the value of $d$ and the value $k$ of the number of regenerators that can be used in any single node.
\begin{itemize}
\item $\prb$: $k=\infty$, $G$ and $d$ are arbitrary (in this case there is a solution for every input, and the measurement is the number of locations in which regenerators are placed). We show:
\begin{itemize}
\item an $O(\log \abs{X} \cdot \log d)$-competitive randomized algorithm for any network topology, that can be made deterministic (with the same competitive ratio) for some cases including tree topology networks, where $X$ is the set of all paths of length $d$ in $G$.
\item a deterministic lower bound of $\Omega \lb$, where $E$ is the edge set of $G$.
\end{itemize}

\item $\prbmax$: $G$ is a path, $k=1$ and $d=2$ (in this case there is not necessarily a solution, and the measurement is the number of satisfied lightpaths). We distinguish between feasible inputs (for which there is a solution) and infeasible ones, on a path topology, and show:
\begin{itemize}
  \item a lower bound of $\sqrt{l}/2$ for the competitive ratio for general instances which may be infeasible (where $l$ is the number of internal nodes of the longest lightpath).
  \item a tight bound of $3$ for the competitive ratio of deterministic online algorithms for feasible instances.
\end{itemize}

\end{itemize}

\runningtitle{Organization of the paper}
In Section~\ref{sec:prelim} we present some preliminaries. In Section~\ref{sec:kinfty} we consider general topology and analyze the first extreme case ($k$ unbounded). In Section~\ref{sec:pathktwo} we analyze the other extreme case ($k=1$) for a path topology. In Section~\ref{sec:summary} we present further research directions.

\section{Preliminaries}\label{sec:prelim}

\noindent
Given an undirected underlying graph $G=(V,E)$ that corresponds to the network topology, a
\emph{lightpath} is a simple path in $G$. We are given a set $\pp=\set{P_1, P_2, ..., P_n}$ of simple paths in $G$ that represent the lightpaths. The \emph{length} of a lightpath is the number of edges it contains. The \emph{internal vertices} (resp. \emph{edges}) of
a path $P$ are the vertices (resp. edges) in $P$ except the first and the last ones.

A regenerator assignment is a function $\reg: V \times \pp \mapsto \set{0,1}$. For any $P \in \pp$ and any $v \in V(P)$, $\reg(v,P)=1$ if a regenerator is assigned to $P$ at node $v$. Moreover, $\reg(v,P)=1$ only if $v$ is an internal node of $P$. We denote by $\reg(v)$ the number of regenerators located at node $v$, i.e., $reg(v) = \sum_{P \in \pp}\reg(v,P)$ .
Denote by $\cost(\reg)$ the cost of the assignment $\reg$, measured by the total number of locations where regenerators have been placed. Let $R(\reg) = \{v \in V | \reg(v) \geq 1\}$, then $\cost(\reg) = |R(\reg) |$.

Given an integer $d$, a lightpath $P$ is \emph{$d$-satisfied} by the regenerator assignment $\reg$ if it does not contain $d$ consecutive internal vertices
without a regenerator,
in other words, for any $d$ consecutive internal vertices of $P$,
$v_1,v_2,\cdots,v_d$, $\sum_{i=1}^d \reg(v_i,P)\geq 1$.
A set of lightpaths is \emph{$d$-satisfied} if each of its lightpaths is $d$-satisfied. Note that a path with at most $d$ edges is $d$-satisfied regardless of $\reg$, therefore we assume without loss of generality that every path $P \in \pp$ has at least $d+1$ edges. For the sake of the analysis we assume, without loss of generality, that every edge of the graph is used by at least one path $P \in \pp$. We want to emphasize that this is not assumed by the online algorithms, (which would be a loss of generality).

The Regenerator Location Problem ($\prb$): given a graph $G=(V,E)$, a set $\pp$ of paths in $G$, a distance $d \geq 1$, determine the smallest number of nodes $R \subseteq V$ to place regenerators so that all the paths in $\pp$ are $d$-satisfied. Formally:

\begin{center}
\fbox{\begin{minipage}{11.cm}
\noindent  {\sc Regenerator Location Problem (\prb)\footnote{The offline version of this problem is denoted as \textsc{RPP}/$\infty$/+ in \cite{FMMMZ11}.}}\\\vspace{-.15cm}

{\bf Input:} An undirected graph $G=(V,E)$, a set $\pp$ of paths in $G$, $d \geq 1$

{\bf Output:} A regenerator assignment $\reg$ such that every path $P \in \pp$ is $d$-satisfied.

{\bf Objective:} Minimize $\cost(\reg)$.
\end{minipage}}
\medskip
\end{center}
Let $\reg^*$ denote an optimal regenerator assignment and
$\cost^*$ denote its cost $\cost(\reg^*)$. We consider the online version of the problem in which $G$ and $d$ are given in advance and the paths $\pp=\set{P_1, P_2, \ldots, P_n}$ arrive in an online manner, one at a time in this order. An online algorithm finds a regenerator assignment as the input arrives and once $\reg(v,P)$ is set to $1$ it cannot be reverted to $0$. An online algorithm~$\alg$ for $\prb$ is \emph{$c$-competitive}, for $c \geq 1$, if its cost is at most $c \cdot \cost^*$. Clearly, when $d=1$, $\cost(\reg)=\abs{V_I}$ for any regenerator assignment $\reg$ where $V_I$ is the set of nodes that are internal nodes of some lightpaths, therefore any algorithm is $1$-competitive. Hence we consider the case $d \geq 2$.

When the number $k$ of regenerators per node is finite, we study the Path Maximization Problem ($\prbmax$): given a graph $G=(V,E)$, a set $\pp$ of paths in $G$, a distance $d \geq 1$ and an integer $k \geq 1$, place regenerators so that the number of $d$-satisfied paths in $\pp$ is maximized. Formally:
\begin{center}
\fbox{\begin{minipage}{11.cm}
\noindent  {\sc Path Maximization Problem (\prbmax)}\\\vspace{-.15cm}

{\bf Input:} An undirected graph $G=(V,E)$, a set $\pp$ of paths in $G$, $d,k \geq 1$

{\bf Output:} A regenerator assignment $\reg$ for which $\reg(v) \leq k$ for every node $v \in V$.

{\bf Objective:} Maximize the number of $d$-satisfied paths in $\pp$.
\end{minipage}}
\medskip
\end{center}

An online algorithm~$\alg$ for $\prbmax$ is \emph{$c$-competitive}, for $c \geq 1$, if the number of paths it satisfies is at least $1/c$ times the number of paths satisfied by an optimal offline algorithm.

\section{The Regenerator Location Problem}\label{sec:kinfty}
In this section we consider the case where the technological limit imposed on the number of regenerators at a node is unlikely to be reached by any regenerator assignment. In this case we can assume without loss of generality that whenever there is a node $v$ and a path $P \in \pp_v$ with $\reg(v,P)=1$ then $\reg(v,P')=1$ for every other path $P' \in \pp_v$, because this does not affect $\cost(\reg)$. In other words for any given node $v$ and any two paths $P,P' \in \pp_v$ we assume $\reg(v,P)=\reg(v,P')$.

\subsection{Bounds for Path Topology}
In this section we consider path topologies and present bounds for the competitive ratio of both deterministic and randomized algorithms. Specifically, we prove a tight bound of $2$ for deterministic algorithms, a lower bound of $3/2$ and an upper bound of $2-1/d^2$ for randomized algorithms.

Throughout the section $V=\set{v_1, v_2, \ldots, v_n}$ is the node set of the path and $E=\set{\set{v_i,v_{i+1}}|1 \leq i < n}$ is its edge set. A \emph{region} of $\pp$ is a maximal set of consecutive vertices in the union of all the internal vertices of the paths of $\pp$. For a region $L$, we denote by $\cost^*_L$ the number of locations from $L$ used by an optimal solution, i.e. $\cost^*_L \defined \abs{\reg^* \cap L}$. We start with the following lower bound regarding an optimal solution.

\begin{proposition}\label{prop:OptPathRegion}
\[
\abs{L} \leq \cost^*_L \cdot (2d-1).
\]
\end{proposition}
\begin{proof}
In any solution the distance between two consecutive regenerator locations from $L$ is at most $2d-1$. Indeed, assume by way of contradiction, that there are $2d-1$ consecutive nodes $u_1,\ldots,u_d,\ldots,u_{2d-1}$ without a regenerator, and consider a path $P$ crossing $u_d$. As $P$ has at least $d$ internal nodes, $P$ crosses either one of $u_1, u_{2d-1}$. Then $P$ has $d$ consecutive nodes without regenerator, i.e. $P$ is not $d$-satisfied, a contradiction. By the same argument, any solution must have a regenerator in at least one of the leftmost (resp. rightmost) $d$ nodes of $L$.
\qed
\end{proof}

\begin{lemma}
\label{lem:two_comp_path}
There is a $2$-competitive deterministic online algorithm in path topologies for $\prb$.
\end{lemma}
\begin{proof}
We set $R=\set{v_d, v_{2d}, \ldots} \subseteq V$ and start with the empty assignment, i.e. $\reg(v)=0$ for every node $v \in V$. When a path $P$ is presented to the algorithm we set $\reg(v)=1$ for every $v \in R \cap P$ which is not an endpoint of $P$. This strategy clearly $d$-satisfies all the paths.

We now show that this algorithm is $2$-competitive. Consider a region $L$ of $\pp$. Clearly, the number $\cost_L$ of locations from $L$ used by our algorithm is at most $\left\lceil \frac{\abs{L}}{d} \right\rceil$. Combining with  Proposition \ref{prop:OptPathRegion} we get
\[
\cost_L \leq \left\lceil \frac{\abs{L}}{d} \right\rceil \leq \frac{\cost^*_L (2d-1)}{d} + 1=2 \cdot \cost^*_L + 1 - \frac{\cost^*_L}{d}
\]
As both $\cost_L$ and $\cost^*_L$ are positive integers, the above implies that $\cost_L \leq 2 \cdot \cost^*_L$. Summing up for all regions, the lemma follows.
\qed
\end{proof}

\begin{lemma}
For every $d \geq 2$, the competitive ratio of every deterministic online algorithm for $\prb$ is at least $2$, even when $G$ is a path.
\end{lemma}
\begin{proof}
The adversary first presents a path $P_0$ with $d+1$ edges. $P_0$ can be $d$-satisfied using one regenerator. If the algorithm uses two regenerators, then the competitive ratio is $2$ and we are done. Therefore we assume that the algorithm uses one regenerator in some internal node $v$ of $P_0$. Then, the adversary presents the path $P_1$ of length $d+1$, having one endpoint in $v$. The other endpoint is chosen (among the two possible nodes) such that the intersection of $P_0$ and $P_1$ is maximized. The algorithm has to use at least one additional regenerator to $d$-satisfy $P_1$, i.e. it uses at least two regenerators in total.

The intersection of the two paths $P_0$ and $P_1$ is at least $2$ edges because $d+1 \geq 3$ and $P_1$ is chosen such that the intersection is maximized. Then, the union of $P_0$ and $P_1$ is a path $P'$ with at most $2(d+1)-2=2d$ edges. Therefore, $P_0$ and $P_1$ can be $d$-satisfied by placing one regenerator at the center of $P'$.
\qed
\end{proof}

The following lemma shows that the above bound does not hold for randomized algorithms.

\begin{lemma}
For every $d \geq 2$, there is a $(2-1/d^2)$-competitive randomized online algorithm in path topologies for $\prb$.
\end{lemma}
\begin{proof}
We choose an integer $i$ between $1$ and $d$ uniformly at random. We set $R=\set{v_{i+d}, v_{i+2d}, \ldots} \subseteq V$ and start with the empty assignment, i.e. $\reg(v)=0$ for every node $v \in V$. When a path $P$ is presented to the algorithm we set $\reg(v)=1$ for every $v \in R \cap P$ which is not an endpoint of $P$. This strategy clearly $d$-satisfies all the paths.

Consider a region $L$ of $\pp$. The bound proven for the deterministic algorithm in Lemma~\ref{lem:two_comp_path} clearly holds for this algorithm too, i.e.
$\cost_L \leq 2 \cdot \cost^*_L$. However, whenever the $d$-th node of $L$ is in $R$ (which happens with probability $1/d$) we can prove a better performance. In this case $\abs{L} \geq \cost_L \cdot d$, because there are $d$ edges  before the first regenerator, $d$ edges between the first two regenerators, etc. Therefore, recalling also Proposition \ref{prop:OptPathRegion}, we have
$\cost_L \leq \abs{L}/d \leq \cost^*_L \cdot (2d-1)/d$. The expectance of $\cost_L$ is at most
\[
E[\cost_L] \leq \frac{1}{d} \cdot \frac{2d-1}{d} \cdot \cost^*_L \cdot + \left(1- \frac{1}{d} \right) 2 \cdot \cost^*_L=\left(2 - \frac{1}{d^2} \right) \cdot \cost^*_L.
\]
The result follows from the linearity of expectance.
\qed
\end{proof}

We conclude this section with a lower bound for randomized algorithms.

\begin{lemma}
For every $d \geq 2$, the competitive ratio of every randomized online algorithm for $\prb$ is at least $3/2$ even when $G$ is a path.
\end{lemma}
\begin{proof}
Using Yao's principle \cite{BE98}, we give an adversary that presents a randomized input and show that the expected cost of any deterministic algorithm on this input is at least $3/2 \cdot \cost^*$.

Consider some path $P_1$ of length $d+1$, and the two paths $P_{21}$ and $P_{22}$ of length $d+1$ having exactly two edges (i.e. one internal node) in common with $P_1$. With probability $1/2$ the adversary presents the input $\pp_1 = \set{P_1, P_{21}}$ and with probability $1/2$ it presents the input $\pp_2 = \set{P_1, P_{22}}$. We note that for $j \in \set{1,2}$ the input $\pp_j$ can be $d$-satisfied using one regenerator, namely the common internal node of $P_1$ and $P_{2i}$.

Any deterministic algorithm that uses $2$ locations to satisfy the path $P_1$ is $2$-competitive. Therefore we assume that the algorithm uses one regenerator at some internal node $v$ of $P_1$. $v$ can $d$-satisfy at most one of $P_{21},P_{22}$. Therefore with probability at least $1/2$ the algorithm uses a second location. The expected number of locations used by the algorithm is at least $2 \cdot 1/2 + 1 \cdot 1/2 = 3/2$.
\qed
\end{proof}

\newcommand{\algrandompath}{\textsc{RandomizedPath}}
\newcommand {\cs}   {{\cal S}}

\subsection{Upper Bound for General Topologies}
In this section we use the randomized algorithm presented in \cite{AAABN09} for the online set-cover problem. For completeness, we provide brief descriptions of the problem and the algorithm.

\newcommand{\cc}{{\cal C}}
\newcommand{\instsc}{(X, \cs)}
\newcommand{\instosc}{(X, \cs, X')}
An instance of the set cover problem is a pair $\instsc$ where $X=\set{x_1, x_2, \ldots}$ is a ground set of elements, and $\cs=\set{S_1, S_2, \ldots}$ is a collection of subsets of $X$. Given such an instance, one has to find a subset $\cc \subseteq \cs$ that covers $X$, i.e. $\cup_{S_i \in \cc} S_i = X$. In \cite{AAABN09} an online variant of the set cover problem is considered. An instance of the online set cover problem is a triple $\instosc$ where $X$ and $\cs$ are as before, and $X' \subseteq X$ is presented in an online manner, one element at a time. At any given time one has to provide a cover $\cc' \subseteq \cs$ of $X'$, i.e. $X' \subseteq \cup_{S_i \in \cc'}S_i$. Once a set is included in the cover $\cc'$ this decision cannot be changed when the subsequent input is received. In other words, whenever an element is presented an online algorithm has to cover it by at least one set from $\cs$ if it is not already covered. It is important to note that $X$ and $\cs$ are known in advance but $X'$ is given online.

\newcommand{\algsc}{\textsc{OnLineSetCover}}
We proceed with a description of the online algorithm in \cite{AAABN09}. We denote by $S^{(i)}$ the set of all sets containing $x_i$, i.e. $S^{(i)} \defined \set{S_j \in \cs | x_i \in S_j}$. Let $f$ be an upper bound for the frequencies of the elements, i.e. $\forall x_i \in X, \abs{S^{(i)}} \leq f$. The algorithm associates a weight $w_j$ with each set $S_j$ which is initiated to $1/f$. The weight $w^{(i)}$ of each element $x_i \in X$ is the sum of the weights of the sets containing it, i.e. $w^{(i)} = \sum_{S_j \in S^{(i)}} w_j$. See pseudo-code in Algorithm \algsc~below for a description of the algorithm.
\alglanguage{pseudocode}
\begin{algorithm}
\caption{\algsc}
\label{alg:algonlinesc}
\begin{algorithmic}[1]
\State When a non-covered element $x_i \in X$ is presented:
\State Find the smallest non-negative integer $q$ such that $2^q \cdot w^{(i)} \geq 1$;
\For{each set $S_j \in S^{(i)}$}
\State $\delta_j = 2^q \cdot w_j - w_j$;
\State $w_j += \delta_j$;
\EndFor
\State \textbf{do} $4 \log |X|$ times
\State ~~~~Add at most one set (from $S^{(i)}$) to the cover\\
~~~~~~~~where each set $S_j$ is chosen with probability $\delta_j/2$;
\end{algorithmic}
\end{algorithm}

From an instance $\instrlp$ of $\prb$ we build an instance $\instosc$ of the online set cover problem. $X$ is the set of all possible paths of length $d$ in $G$ and $\abs{\cs}=\abs{V}$. Each set $S_j \in \cs$ consists of all the paths in $X$ containing the node $v_j$. For a path $P$, let $P^{(d)}$ be the set of all its sub-paths of length $d$. $X'$ is $\cup_{P \in \pp} P^{(d)}$. Now we observe that for any feasible regenerator assignment $\reg$, $R(\reg)$ corresponds to a set cover, and vice versa, i.e. any set cover $\cc$ corresponds to a feasible regenerator assignment $\reg$ such that $R(\reg)=\cc$. Indeed, a path $P$ is $d$-satisfied if and only if every path of $P^{(d)} \subseteq X'$ contains a node $v_j$ with regenerators, that corresponds to a set $S_j \in \cc$ containing this path. Therefore all the paths $P \in \pp$ are $d$-satisfied if and only if $\cc$ constitutes a set cover of $X'$. Moreover the cost of the set cover is equal to the number of regenerator locations, i.e. $\abs{\cc}=\sum_{v_j} \reg(v_j) = \cost(\reg)$.

When a path $P$ is presented, we present to \algsc~all the paths of $P^{(d)}$ one at a time. For each set $S_j$ added to the cover by \algsc, we set $\reg(v_j)=1$.

We first note that although the number of sets in $X$ is exponential in terms of the input size of our problem, for every path $P$ the set $P^{(d)}$ contains only a polynomial number of paths, therefore the first loop of Algorithm \algsc~ runs only a polynomial number of times. The second loop is executed $\log \abs{X}$ times, which is also polynomial in terms of our input size.

Algorithm \algsc~is proven to be $O(\log \abs{X} \cdot \log f)$-competitive. Note that a path of length $d$ contains $d+1$ nodes, thus $f=d+1$. As the cost of a cover is equal to the cost of a solution of $\instrlp$ we conclude
\begin{lemma}
There is an $O(\log \abs{X} \cdot \log d)$-competitive polynomial-time randomized online algorithm for instances $\instrlp$ of $\prb$ where $X$ is the set of all the paths of length $d$ in $G$.
\end{lemma}

In \cite{AAABN09} algorithm \algsc~is de-randomized using the method of conditional expectation. However in this method, in order to calculate the conditional expectancies, one has to consider all the elements of $X$. In our case $X$ is the set of all paths of length $d$ in $G$ which is, in general, exponential in $d$, thus applying the technique in             \cite{AAABN09} directly to our case leads to an exponential algorithm.  Although the definition of competitive ratio does not require polynomial running-time, for practical purposes we would like to have polynomial-time algorithms. The following theorem states some cases for which this condition is satisfied.
\begin{theorem}
There is an $O(\log \abs{X} \cdot \log d)$-competitive polynomial-time deterministic online algorithm for instances $\instrlp$ of $\prb$ in each one of the following cases where $X$ is the set of all the paths of length $d$ in $G$.
\begin{itemize}
\item Both $d$ and the maximum degree $\Delta(G)$ of $G$ are bounded by some constant.
\item The number of cycles in $G$ is bounded, in particular $G$ is a ring.
\item $G$ has bounded treewidth, in particular $G$ is a tree.
\end{itemize}
\end{theorem}

\comment{Proof in journal version}

\subsection{Lower Bound for General Topologies}\label{subsec:lower bounds}
In this section we show a lower bound nearly matching the upper bound in the previous subsection, by using the online version of a reduction in \cite{FMMMZ11} of set cover to $\prb$. Given an instance $\instosc$ of online set cover we build an instance $\instrlp$ of $\prb$ as follows (see Figure \ref{fig:general}).

\begin{figure}[h]\label{fig:reduction}
\begin{center}
\scalebox{0.8}{\begin{picture}(0,0)%
\special{psfile=reduction.pstex}%
\end{picture}%
\setlength{\unitlength}{2901sp}%
\begingroup\makeatletter\ifx\SetFigFont\undefined
\def\x#1#2#3#4#5#6#7\relax{\def\x{#1#2#3#4#5#6}}%
\expandafter\x\fmtname xxxxxx\relax \def\y{splain}%
\ifx\x\y   
\gdef\SetFigFont#1#2#3{%
  \ifnum #1<17\tiny\else \ifnum #1<20\small\else
  \ifnum #1<24\normalsize\else \ifnum #1<29\large\else
  \ifnum #1<34\Large\else \ifnum #1<41\LARGE\else
     \huge\fi\fi\fi\fi\fi\fi
  \csname #3\endcsname}%
\else
\gdef\SetFigFont#1#2#3{\begingroup
  \count@#1\relax \ifnum 25<\count@\count@25\fi
  \def\x{\endgroup\@setsize\SetFigFont{#2pt}}%
  \expandafter\x
    \csname \romannumeral\the\count@ pt\expandafter\endcsname
    \csname @\romannumeral\the\count@ pt\endcsname
  \csname #3\endcsname}%
\fi
\fi\endgroup
\begin{picture}(7625,4266)(164,-3658)
\put(6551,389){\makebox(0,0)[lb]{\smash{\SetFigFont{9}{10.8}{rm}$S_m$}}}
\put(189,-461){\makebox(0,0)[lb]{\smash{\SetFigFont{9}{10.8}{rm}$s_1$}}}
\put(201,-985){\makebox(0,0)[lb]{\smash{\SetFigFont{9}{10.8}{rm}$s_2$}}}
\put(214,-1535){\makebox(0,0)[lb]{\smash{\SetFigFont{9}{10.8}{rm}$s_3$}}}
\put(189,-2848){\makebox(0,0)[lb]{\smash{\SetFigFont{9}{10.8}{rm}$s_{n-1}$}}}
\put(164,-3598){\makebox(0,0)[lb]{\smash{\SetFigFont{9}{10.8}{rm}$s_n$}}}
\put(1764,401){\makebox(0,0)[lb]{\smash{\SetFigFont{9}{10.8}{rm}$S_1$}}}
\put(2551,401){\makebox(0,0)[lb]{\smash{\SetFigFont{9}{10.8}{rm}$S_2$}}}
\put(3551,401){\makebox(0,0)[lb]{\smash{\SetFigFont{9}{10.8}{rm}$S_3$}}}
\put(7771,-1024){\makebox(0,0)[lb]{\smash{\SetFigFont{9}{10.8}{rm}$t_2$}}}
\put(1226,-136){\makebox(0,0)[lb]{\smash{\SetFigFont{9}{10.8}{rm}$v_{1,1}$}}}
\put(2489,-1024){\makebox(0,0)[lb]{\smash{\SetFigFont{9}{10.8}{rm}$v_{2,2}$}}}
\put(7789,-361){\makebox(0,0)[lb]{\smash{\SetFigFont{9}{10.8}{rm}$t_1$}}}
\put(7751,-1724){\makebox(0,0)[lb]{\smash{\SetFigFont{9}{10.8}{rm}$t_3$}}}
\put(7701,-2724){\makebox(0,0)[lb]{\smash{\SetFigFont{9}{10.8}{rm}$t_{n-1}$}}}
\put(7689,-3436){\makebox(0,0)[lb]{\smash{\SetFigFont{9}{10.8}{rm}$t_n$}}}
\put(5626,389){\makebox(0,0)[lb]{\smash{\SetFigFont{9}{10.8}{rm}$S_{m-1}$}}}
\end{picture}}
\caption{Reduction from online set cover to \prb}
\label{fig:general}
\end{center}
\end{figure}
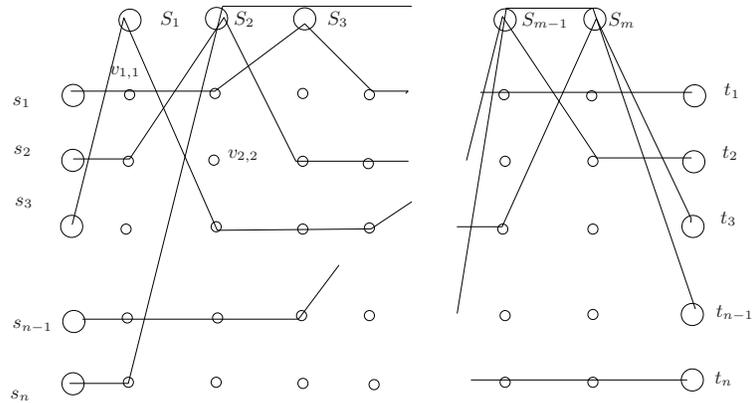

We set $d=\abs{\cs}$. The node set $V(G)$ of $G$ is $\cs \cup V_1 \cup V_2$ where $V_1=\set{s_i, t_i | 1 \leq i \leq \abs{X}}$ and $V_2=\set{v_{ij} | 1 \leq i \leq \abs{X}, 1 \leq j \leq \abs{\cs}}$. We proceed with a description of the paths $\pp$. The edge set of $G$ will be all the edges induced by the paths of $\pp$. For each element $x_i$ there is a path $P_i$ in $\pp$ between $s_i$ and $t_i$. If $x_i \in S_j$ then $S_j \in V(G)$ is an internal node of $P_i$, otherwise $v_{ij}$ is an internal node of $P_i$. The internal nodes are ordered within the path $P_i$ by their $j$ index, i.e. the path $x_i$ is of the form $(s_i-u_1-u_2-\cdots-u_{\abs{\cs}}-t_i)$ where $u_j$ is either $S_j$ or $v_{ij}$ as described before.

By this construction every path $x_i$ has exactly $\abs{\cs}=d$ internal nodes. Therefore a regenerator assignment is feasible if and only if it assigns at least one regenerator to one of the internal nodes of every path. Without loss of generality every element $x_i$ is contained in at least one set $S_j$, otherwise no set cover exists. A feasible regenerator assignment $\reg$ corresponds to a set cover, in the following way. We first obtain a regenerator assignment $\reg'$ such that $\reg'(v_{ij})=0$ for every $v_{ij} \in V_2$ and $\cost(\reg') \leq \cost(\reg)$. For  every node with $\reg(v_{ij})=1$ we set $\reg'(v_{ij})=0$, and if $P_i$ is not $d$-satisfied in $\reg'$ we choose arbitrarily a node $S_j$ on $P_i$ and set $\reg'(S_j)=1$. Now $R(\reg') \subseteq \cs$ is a set cover of cardinality at most $\cost(\reg)$.

\renewcommand{\alg}{\textsc{Alg}}
\begin{lemma}
There is no $O(\frac{\nominator}{\log (\nominator)})$-competitive online algorithm for $\prb$.
\end{lemma}
\begin{proof}
Assume by contradiction that there is an $O(\frac{\nominator}{\log (\nominator)})$-competitive algorithm $\alg$ for $\prb$. From an instance $\instosc$ of online set cover we build an instance of $\prb$ as described in the above discussion, and whenever we are presented an element $x_i \in X' \subseteq X$ we present the path $P_i$ to \alg. We transform the regenerator assignment returned by \alg~to a set cover $\cc$ as described above. Note that the transformation does not exclude a set $S_j$ from $\cc$ if is was already in $\cc$ before $x_i$ was presented, thus $\cc$ is an online set cover.
\newcommand{\nominatorsc}{\log \abs{X} \cdot \log \abs{\cs}}
We note that $\abs{V}=\Theta(\abs{X} \cdot \abs{\cs}), \abs{E}=\Theta(\abs{V}), d=\Theta(\abs{\cs})$. This implies an $O(\frac{\nominatorsc}{\log (\nominatorsc)})$-competitive algorithm for the online set cover problem, which is proven to be impossible in \cite{AAABN09}.
\qed
\end{proof}

\section{Path Maximization in Path Topology. ($k=1, d=2$)}\label{sec:pathktwo}
In this section we consider possibly the simplest instances of the $\prbmax$ problem, i.e. the case where the network is a path, and $k=1$, $d=2$.

We say that an instance is \emph{feasible}, if there is a regenerator assignment that $d$-satisfies all the paths in $\pp$, and \emph{infeasible} otherwise. We first show in Section~\ref{sec:koneinfeasible} that if the input instance is infeasible,
no online algorithm (for $\prbmax$) has a small competitive ratio; precisely, we show that no online algorithm
is better than $\sqrt{l}/2$-competitive, where $l$ is the length of the longest path in the input.
We then focus on feasible instances in Section~\ref{sec:konefeasible}.

\subsection{Infeasible Instances}
\label{sec:koneinfeasible}
We show that there is a lower bound in terms of the length of the longest path
if the input instance is infeasible, as follows:

\begin{lemma}
\label{thm:koneinfeasible}
Consider the path topology. For $k=1$ and $d=2$,
any deterministic online algorithm for $\prbmax$ has a competitive ratio at least $\sqrt{l}/2$,
where $l$ is the number of internal vertices of the longest path.
\end{lemma}

\begin{proof}
The adversary first releases a path of length $l+1$ with $l$ internal vertices.
The online algorithm has to satisfy this path,
otherwise, the competitive ratio is unbounded.
Then the adversary releases $\sqrt{l}$ paths along the first path
each with $\sqrt{l}$ (disjoint) internal vertices.
If the online algorithm does not satisfy any of these paths,
the competitive ratio is at least $\sqrt{l}$ and we are done.
Suppose $x$ of these paths are satisfied.
In order to make the first path and these $x$ paths
$2$-satisfied,
there is one regenerator placed in each node along
these $x$ paths.
For each of these $x$ paths $P$,
the adversary releases $\sqrt{l}/2$ paths along $P$
each with two (disjoint) internal vertices.
The online algorithm is not able to satisfy any of these short paths
and the total number of $2$-satisfied paths is $x+1$.
On the other hand,
the optimal offline algorithm satisfies all the paths except the first
path of length $l$, i.e., $\sqrt{l} + x\sqrt{l}/2$ paths.
As a result, the competitive ratio of the online algorithm is
$\frac{(x+2)\sqrt{l}}{2(x+1)} > \sqrt{l}/2$.
\qed
\end{proof}

\subsection{Feasible Instances}
\label{sec:konefeasible}
We now consider feasible instances,
that is, instances, where there exists a placement
of regenerators such that all paths are satisfied.
We will prove that, for feasible instances,
there is a tight bound of $3$ for the competitive ratio.
That is, we provide an online algorithm Algorithm~\ref{online-feasible} with competitive ratio
$3$, and we show a lower bound
of $3$ for the competitive ratio of every deterministic online algorithm
for feasible instances.

Note that a regenerator assignment $2$-satisfies a path $P$ if and only if it constitutes a vertex cover of the edges of $P$, except its first and last edges. Therefore, in this section, for simplicity we assume that the leftmost and rightmost edges of the paths have been removed and a regenerator assignment is a vertex cover of the edges of the paths.

Algorithm~\ref{online-feasible} adopts a greedy approach and satisfies
a newly presented path whenever possible.
When a path $P_i$ is presented,
it checks whether there exist two consecutive internal vertices of $P_i$ that
are already assigned regenerators for previous paths.
If yes, this means it is impossible (under the current assignment) to satisfy $P_i$.
Otherwise, the algorithm satisfies $P_i$, as follows. There are two possible locations for the leftmost regenerator of $P_i$, namely, either its leftmost internal node, or the internal node adjacent to it. Among these two alternatives
we choose the alternative that uses the smaller number of regenerators
by trying the following regenerator allocation process.
Suppose we put a regenerator at a certain internal node $v$ of $P_i$.
We check whether the node at distance 2 from $v$ already has a regenerator;
if not, we put a regenerator there and continue;
if yes, we put a regenerator at the node at distance 1 from $v$
\footnote{The node at distance 1
must have no regenerator, else there are two consecutive internal nodes
with regenerators and the algorithm would have rejected the path.}.
This continues until $P_i$ is $2$-satisfied.

\alglanguage{pseudocode}
\begin{algorithm}
\caption{Online algorithm for a path-topology, $k=1$ and $d=2$.}
\label{online-feasible}
\begin{algorithmic}[1]
\State When the path $P_i$ with endpoints $s_i,t_i$ is presented:
\If{$\reg(v)=\reg(v')=1$ for two consecutive internal nodes $v,v'$ of $P_i$}
\State leave $P_i$ unsatisfied;
\Else
\State $last \gets s_i$
\While {$P_i$ is not $2$-satisfied}
\If {$\reg(last+2)=0$}
\State $last \gets last+2$
\Else
\State $last \gets last+1$
\EndIf
\State $\reg(last,P_i) \gets 1$
\EndWhile
\EndIf
\end{algorithmic}
\end{algorithm}

\begin{theorem}\label{online-feasible-comp-ratio-3-lem}
Algorithm~\ref{online-feasible} is $3$-competitive for $\prbmax$ for feasible inputs in path topologies, when $k=1$ and $d=2$.
\end{theorem}

\begin{proof}
Let $S$ and $U$ denote the sets of paths that have been satisfied and
unsatisfied by the algorithm, respectively.
We prove the theorem by showing that $|U|\leq 2|S|$.
Then, the competitive ratio of Algorithm~\ref{online-feasible} is $\frac{|\pp|}{|S|}=\frac{|U|+|S|}{|S|} \leq \frac{2|S|+|S|}{|S|} = 3$, i.e., Algorithm~\ref{online-feasible} is $3$-competitive.
In the sequel we prove that $|U|\leq 2|S|$ by associating with every path in $U$ some paths of $S$,
and showing that each path in $S$ is associated with at most two paths in $U$.

Note also that, since the instance is assumed to be feasible,
for every edge $uv$ there exist \emph{at most} two paths $P_i,P_j$,
such that $uv\in P_i$ and $uv\in P_j$
(indeed, otherwise there would exist at least one path
that is unsatisfied on the edge $uv$).
Suppose that a path $P_i$ presented at iteration $i$ is unsatisfied, i.e., when $P_i$ arrives, it cannot be satisfied by placing new regenerators.
Then, there exists an edge $ab\in P_i$,
where both $a$ and $b$ already have regenerators of paths that have been previously satisfied by the algorithm.
We distinguish now two cases regarding the regenerators on vertices $a$ and $b$.

\vspace{3mm}\noindent{\bf Case 1:}
The regenerators on vertices $a$ and $b$ belong to two different paths $P_j$
and $P_h$ which have been satisfied previously by the algorithm,
i.e., $\reg(a,P_j)=\reg(b,P_h)=1$,
with $j,h<i$ and $j\neq h$.

We first consider the cases where $ab \in P_j$ or $ab \in P_h$.
Suppose that $ab\in P_j$. Then, since also $ab\in P_i$ by assumption, it follows that $ab\notin P_h$, since the instance is feasible. That is, $b$ is an endpoint of $P_h$. In this case, associate the unsatisfied path $P_i$ to the satisfied path $P_h$.
Suppose now that $ab\in P_h$.
Then it follows similarly that $ab\notin P_j$,
and thus $a$ is an endpoint of $P_j$.
In this case, associate the unsatisfied path $P_i$ to the satisfied path $P_j$.

Suppose now that $ab\notin P_j$ and $ab\notin P_h$, i.e., $a$ is an endpoint of $P_j$ and $b$ is an endpoint of $P_h$. If there exists another path $P_{\ell}$ that is left unsatisfied by the algorithm, such that $ab\in P_{\ell}$, then associate the unsatisfied paths $\{P_i,P_{\ell}\}$ to the satisfied paths $\{P_j,P_h\}$. Otherwise, if no such path $P_{\ell}$ exists, then associate the path $P_i$ to either $P_j$ or $P_h$.

\noindent{\bf Case 2:}
The regenerators on $a$ and $b$ belong to the same path $P_j$
which has been satisfied previously by the algorithm,
i.e., $\reg(a,P_j)=\reg(b,P_j)=1$, where $j<i$.

The edge $ab\in P_j$. Furthermore, neither $a$ nor $b$ is an endpoint of path $P_j$, since otherwise Algorithm~\ref{online-feasible} would not place a regenerator on both vertices $a$ and $b$ of path $P_j$. That is, there exist two vertices $d,c$ of $P_j$, such that $(d,a,b,c)$ is a subpath of $P_j$. Moreover, since $a$ and $b$ are consecutive vertices of $P_j$, according to the algorithm there must exist two other satisfied paths $P_h$, $P_{\ell}$, such that $\reg(d,P_h)=\reg(c,P_{\ell})=1$.\footnote{%
Here we simplify the discussion slightly by assuming that
the path $P_i$ does not contain a chain of two internal edges that both
do not belong to any other paths because the algorithm can simply
assign regenerators to alternate internal nodes without conflicting
any other paths and this would not affect the number of paths that can
be satisfied by the algorithm.
}
Note also that $ab\notin P_{h}$ and $ab\notin P_{\ell}$, since the instance is feasible, and since $ab\in P_{i}$
and $ab\in P_{j}$. That is, $d$ or $a$ is an endpoint of $P_h$, while $b$ or $c$ is an endpoint of $P_{\ell}$.

We claim  that there exist \emph{at most} two different unsatisfied paths
$P_i$ and $P_{i'}$ that include at least one of the edges $da,ab,bc$. Suppose otherwise that there exist three such unsatisfied paths $P_i$, $P_{i'}$, $P_{i''}$. Recall that $ab\in P_i$ and that $da,ab,bc\in P_j$. Therefore, since the instance is assumed to be feasible, it follows that, either $da\in P_{i'}$ and $bc\in P_{i''}$, or $bc\in P_{i'}$ and $da\in P_{i''}$. Since these cases are symmetric, we assume without loss of generality that $da\in P_{i'}$ and $bc\in P_{i''}$. In any optimal (i.e., offline) solution, at least one of $\{a,b\}$ has a regenerator for path $P_j$; assume without loss of generality that $reg(b,P_j)=1$ (the other case $reg(a,P_j)=1$ is symmetric). Then, it follows that $reg(a,P_i)=1$. Then, since the edge $da$ must be satisfied for both paths $P_j$ and $P_{i'}$, it follows that $reg(d,P_j)=reg(d,P_{i'})=1$. This is a contradiction, since every vertex can have at most one regenerator. Therefore there exist at most two different unsatisfied paths $P_i$, $P_{i'}$ that include at least one of the edges $da,ab,bc$.

In the case that $P_i$ is the only unsatisfied path that includes at least one of the edges $da,ab,bc$, associate the unsatisfied path $P_i$ to either the satisfied path $P_h$ or to the satisfied path $P_{\ell}$. Otherwise, if there exist two different unsatisfied paths $P_i$, $P_{i'}$ that include at least one of the edges $da,ab,bc$, associate the unsatisfied paths $\{P_i,P_{i'}\}$ to the satisfied paths $\{P_h,P_{\ell}\}$.

\vspace{5mm}
We observe that by the above associations of unsatisfied paths to satisfied ones, that at most two unsatisfied paths are associated to every satisfied path $P$ (i.e., at most one to the left side and one to the right side of $P$, respectively).
This gives $|U|\leq 2|S|$ and the theorem follows.
\qed
\end{proof}

\begin{lemma}\label{online-feasible-lower-bound-3-lem}
Any deterministic online algorithm for $\prbmax$ has a competitive ratio at least $3$ even when the instance is restricted to feasible ones on path topologies and $k = 1, d = 2$.
\end{lemma}

\begin{proof}
We will prove that, for every $\varepsilon > 0$, there exist infinitely many inputs such that every algorithm has competitive ratio at least $3-\varepsilon$. Choose an integer $n$, such that $\frac{2}{n+1}<\varepsilon$. The adversary provides initially a path $P_0$ with $13n - 2$ edges. The algorithm must satisfy the path $P_0$, since otherwise the adversary stops and the competitive ratio is infinite. We divide $P_0$ into $n$ subpaths $P_{i}$, $i=1,2,\ldots,n$, with $11$ edges each, where between two consecutive subpaths there exist two edges.

Consider any such subpath $P_{i}$, $i=1,2,\ldots,n$. Suppose that there exist two edges $ab$ and $cd$ of $P_{i}$, where $\{a,b\}\cap\{c,d\}=\emptyset$, such that $\reg(a,P_0)=\reg(b,P_0)=\reg(c,P_0)=\reg(d,P_0)=1$. Then the adversary provides next the paths $P_{i,1}=(a,b)$ and $P_{i,2}=(c,d)$. These two paths $P_{i,1}$ and $P_{i,2}$ cannot be satisfied, since each of the vertices $a,b,c,d$ has a regenerator for path $P_0$. So the competitive ratio of the algorithm is at least $3$.

We thus can assume that there do not exist such edges $ab$ and $cd$ for any of the $P_{i}$'s. That is, there exist \emph{at most} three consecutive vertices $u_1,u_2,u_3$ of $P_i$, such that $\reg(u_1,P_0)=\reg(u_2,P_0)=\reg(u_3,P_0)=1$, while for every other edge $uu'$ of $P_i$, there exists a regenerator for $P_0$ either on vertex $u$ or on vertex $u'$. Then, there exist five consecutive vertices $v_1^i,v_2^i,v_3^i,v_4^i,v_5^i$ of $P_i$, such that $\reg(v_1^i,P_0)=\reg(v_3^i,P_0)=\reg(v_5^i,P_0)=1$ and $\reg(v_2^i,P_0)=\reg(v_4^i,P_0)=0$.

The adversary now provides the path $P_{i}^{\prime}=(v_2^i,v_3^i,v_4^i)$. Thus, since $\reg(v_3^i,P_0)=1$ and $\reg(v_2^i,P_0)=\reg(v_4^i,P_0)=0$, the only way that the algorithm can satisfy $P_{i}^{\prime}$ is to place regenerators for $P_{i}^{\prime}$ at the vertices $v_2^i$ and $v_4^i$ (that is, $\reg(v_2^i,P_{i}^{\prime})=\reg(v_4^i,P_{i}^{\prime})=1$).

The adversary proceeds as follows. In the case where the algorithm chooses not to satisfy the path $P_{i}^{\prime}$, the adversary does not provide any other path that shares edges with $P_{i}$. Otherwise, if the algorithm satisfies $P_{i}^{\prime}$, then the adversary provides the paths $P_{i}^{\prime\prime}=(v_1^i,v_2^i)$ and $P_{i}^{\prime\prime\prime}=(v_4^i,v_5^i)$ (see Figure~\ref{fig:lowerbound}). In this case, $\reg(v_2^i,P_{i}^{\prime})=\reg(v_4^i,P_{i}^{\prime})=1$ and $\reg(v_1^i,P_0)=\reg(v_5^i,P_0)=1$, and thus the paths $P_{i}^{\prime\prime}$ and $P_{i}^{\prime\prime\prime}$ remain unsatisfied by the algorithm.

We now show that this instance is indeed feasible. Actually, we show that even the instance that includes $P_0$, all the  paths $P_{i}^{\prime}=(v_2^i,v_3^i,v_4^i)$, all the paths $P_{i}^{\prime\prime}=(v_1^i,v_2^i)$ and all the paths $P_{i}^{\prime\prime\prime}=(v_4^i,v_5^i)$ is feasible.
To see this, we put regenerators at the nodes $v_1^i,v_3^i,v_5^i$, that will satisfy $P_{i}^{\prime\prime}=(v_1^i,v_2^i)$, $P_{i}^{\prime}=(v_2^i,v_3^i,v_4^i)$ $P_{i}^{\prime\prime\prime}=(v_4^i,v_5^i)$, respectively. We then put regenerators at all other nodes (including the nodes $v_2^i,v_4^i$), which clearly satisfies $P_0$.

\begin{figure}[h]
\begin{center}
\scalebox{1.0}{\input{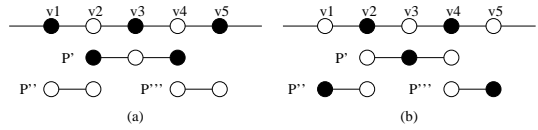}}
\caption{Adversary for Lemma~\ref{online-feasible-lower-bound-3-lem}.
 (a) The online assignment where $P''$ and $P'''$ cannot be satisfied.
 (b) The optimal assignment where all paths are satisfied.}
\label{fig:lowerbound}
\end{center}
\end{figure}

Denote by $h$ the number of subpaths $P_i$, for which the adversary adds the path $P_{i}^{\prime}$, $P_{i}^{\prime\prime}$ and $P_{i}^{\prime\prime\prime}$. Thus the number of subpaths $P_i$, for which the adversary adds $P_{i}^{\prime}$, but not $P_{i}^{\prime\prime}$ or $P_{i}^{\prime\prime\prime}$, is $n-h$. The total number of paths that the adversary provided is thus $1 + 3h + (n-h)~=~1+n+2h$.
The number of paths satisfied by the algorithm is $1 + h$. That is, the competitive ratio of the algorithm is $\frac{1 + n + 2h}{1 + h} ~ =~ 3 + \frac{n - 2 - h}{1 + h}$. Therefore, since $h\leq n$, it follows that the competitive ratio of the algorithm is at least $3 - \frac{2}{1 + n} > 3-\varepsilon$. Since this holds for every $\varepsilon>0$, it follows that any deterministic online algorithm has competitive ratio at least $3$.
This completes the proof of the lemma.
\qed
\end{proof}


\section{Future Work}\label{sec:summary}

We list some open problems and research directions:
\begin{itemize}
\item Close the gap between the bounds shown in this paper. In particular, we used in Section \ref{sec:kinfty} a known approximation result of set cover and modified it for our problem. It might be of interest to improve the upper bound by developing a better algorithm for these special instances of the set cover problem. However we note that \algsc~does not use the set of all potential elements but only its size. Therefore if the algorithm has a priori information about the total length of the paths to be received, the algorithm can use it to get an upper bound which is logarithmic in terms of this bound, instead of the number of all possible paths of size $d$ which can be much bigger.

\item Extend the results for other values of the parameters $d$ and $k$.

\item Consider the regenerator location problem when also traffic grooming is allowed (that is, when up to $g$ (the {\em grooming factor}) paths that share an edge can be assigned the same wavelength and can then share regenerators).
     In \cite{FMMSZ11b} optimizing the use of regenerators in the presence of traffic grooming is studied, but with two fundamental differences: (1) the cost function in this work is the number of regenerator locations rather than the total number of regenerators suggested in that work, and (2) we consider the online case, where the requests for connection are not known a priori, whereas that work considers the offline case where all requests are given in advance.
    \item Consider other objective functions (some of them are discussed in Section~\ref{sec:intro}).
\end{itemize}

\bibliographystyle{abbrv}
\bibliography{Optical,Approximation,Mordo,Miscelaneous}

\end{document}